\documentclass[11pt]{article}
\usepackage{amsfonts,amsmath}
\usepackage{url}
\usepackage{enumitem}
\usepackage{graphicx}
\usepackage{subcaption}
\usepackage{mathrsfs}
\usepackage{bbm}
\usepackage{hyperref}

\newtheorem{theorem}{Theorem}[section]

\newtheorem{definition}[theorem]{\textbf{Definition}}
\newtheorem{remark}[theorem]{\textbf{Remark}}
\newtheorem{lemma}[theorem]{\textbf{Lemma}}
\newtheorem{corollary}[theorem]{\textbf{Corollary}}

\DeclareMathOperator{\vol}{vol}
\DeclareMathOperator{\length}{length}
\DeclareMathOperator{\fatness}{fatness}
\DeclareMathOperator{\vdis}{VDis}

 \newcommand{\R}{\mathbb{R}}
\renewcommand{\H}{\mathbb{H}}
\renewcommand{\S}{\mathbb{S}}

\newenvironment{proof}[1][{}]{
   \begin{trivlist}\item[]\textit{Proof #1}\quad}%
   {\hfill\hspace*{\fill}~$\square$\end{trivlist}}

\let\omarginpar\marginpar
\def\marginpar#1{\omarginpar{\footnotesize#1}}

 \title{\MakeUppercase{Simplicial subdivision of simplices of arbitrary dimension in spaces of constant curvature with bounded quality}\thanks{\scriptsize The research leading to these results has received funding from the European Research Council (ERC) under  the  European  Union's  Seventh  Framework  Programme (FP/2007-2013)  /  ERC  Grant Agreement No. 339025 GUDHI (Algorithmic Foundations of Geometry Understanding in Higher Dimensions). 
\newline
The first author is further supported by the French government, through the 3IA C\^ote d’Azur Investments in the Future project managed by the National Research Agency (ANR) with the reference number ANR-19-P3IA-0002.
The last author is/has been supported by the European Union's Horizon 2020 research and innovation programme under the Marie Sk{\l}odowska-Curie grant agreement No. 754411, the Austrian science fund (FWF) M-3073, ANR grant StratMesh, ANR-24-CE48-1899, and the welcome
 package from IDEX of the Université Côte d’Azur, ANR-15-IDEX-01.}
}

\author{Jean-Daniel~Boissonnat,\thanks{Datashape, Inria centre Universit{\'e} C{\^o}te d'Azur, France,
\texttt{jean-daniel.boissonnat@inria.fr}}\,
Hana~Dal~Poz~Kou\v{r}imsk\'a,\thanks{{University of Potsdam, Germany}, 
\texttt{hana.dal.poz.kourimska@uni-potsdam.de}}\,
\\[0.5em]
Arijit~Ghosh,\thanks{{ACM Unit, Indian statistical institute, Kolkata, India}, 
\texttt{arijitiitkgpster@gmail.com}}\,
and 
Mathijs~Wintraecken\thanks{{Datashape, Inria centre Universit{\'e} C{\^o}te d'Azur, France}, 
\texttt{mathijs.wintraecken@inria.fr}}
}
\date{}

\begin{document}

\maketitle

\begin{abstract}
In 1942, Freudenthal showed that a simplex in Euclidean space can be subdivided such that the quality (well-shapedness of the simplex, quantified in terms of e.g. fatness) of the simplices in the subdivision is lower bounded. This answered a question of Brouwer.  Recently, Brunck discussed the same problem for simplices in two-dimensional spaces of constant curvature and provided a closely related construction.  
In this paper we generalize Brunck's result to arbitrary dimensional spaces of constant curvature by combining Freudenthal's construction and radial projection. We contrast this approach with Brunck's construction. 
\end{abstract}

\section{Introduction}
Subdividing simplices while preserving quality is a longstanding problem, which dates back at least to Brouwer. In 1942, Freudenthal addressed Brouwer's question and showed in~\cite{freudenthal1942simplizialzerlegungen} that simplices in Euclidean space can be subdivided without decreasing the quality. Also within computational geometry the question has a long history, see e.g. \cite{BERN1995}. 
Recently, simplices of good quality have regained attention because of their importance to numerical accuracy in numerical partial differential equations~\cite{shewchuk2002good}.

Simplex subdivision with bounded quality in two-dimensional spaces of constant curvature has recently been investigated by Brunck~\cite{brunck1}. In this paper, we show that using Freudenthal's construction together with the radial projection suffices to prove that simplices in spaces of constant curvature of any dimension can be subdivided with lower bounded quality.

\paragraph{Outline}
In Section~\ref{sec:preliminaries} we introduce the main tools for our theorem --- the notion of quality, the Freudenthal--Kuhn triangulation, and the radial projection. We also
explain Freudenthal's subdivision which is a key ingredient for our subdivision scheme.

In Section~\ref{sec:mainConstruction} we use these tools to establish a subdivision scheme for any simplex in a space of constant curvature that ensures a lower bound on the quality of the simplices within the subdivision. 

Finally, in Section \ref{sec:Comp} we compare our approach to the construction of Brunck~\cite{brunck1}, shedding light on how their construction can be extended from the considered two-dimensional simplices to simplices of arbitrary (finite) dimension. 

We conclude the article with several remarks and open questions (Section~\ref{sec:final_remarks}).

\section{Preliminaries}\label{sec:preliminaries}
The goal of this paper is to show that simplices in spaces of constant curvature can be subdivided while maintaining control over the quality of the simplices in the subdivision.

Here, we view a \emph{simplex} as a convex hull of a set of affinely independent points. 
A \emph{triangulation} of a (topological) space is a simplicial complex that is homeomorphic to that space.
A \emph{subdivision} of a triangulation is a refinement of a given triangulation obtained by subdividing its simplices into smaller simplices, without changing the underlying space.

Several notions of quality have been proposed for simplices.
In this article we rely on a common quality measure called fatness:
\begin{definition}\label{def:fatness}
    The \textbf{fatness} of a $d$-dimensional simplex $\Delta$ is its volume divided by the length of its longest edge $e_\Delta$ to the power $d$:
\[
\fatness(\Delta)=\frac{\vol(\Delta)}{\length(e_\Delta)^d}
\]
\end{definition}

We refer to for example \cite{CHOUDHARY201933, CoxeterQual} for a more general discussion on the quality of simplices in general Coxeter and related arrangements.

The proof of our main theorem is constructive and relies on the combination of the Freudenthal--Kuhn triangulation and the radial projection. We revise these notions next.

\paragraph{Freudenthal--Kuhn Triangulations} 
For historical reasons, the triangulation of the Euclidean space constructed by Freudenthal in~\cite{freudenthal1942simplizialzerlegungen} is referred to as the \emph{Freudenthal--Kuhn triangulation}\footnote{This triangulation and many closely related to it have a long history --- see the bibliography and some historical remarks in \cite{FullVersion, JournalTracing}. Furthermore, Freudenthal's construction is closely related to Coxeter triangulations \cite{FullVersion, JournalTracing, coxeter1934discrete,dobkin1990contour}. In fact, up to an affine transformation, it is a Coxeter triangulation of type $\tilde{A}_d$, where $d$ refers to the dimension.}. 
Roughly speaking, the Euclidean space is divided into unit cubes, and each cube is in turn subdivided into simplices using certain diagonals of the cube and its faces. See Figure~\ref{fig:triangulation}.

\begin{figure}[h!]
	\centering
		\includegraphics[width=.25\textwidth]{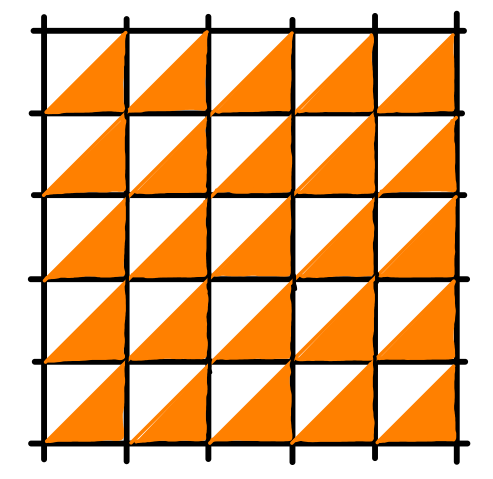}
	\caption{
		Freudenthal--Kuhn triangulation in dimension 2.
	}
	\label{fig:triangulation}
\end{figure}

One of the equivalent definitions is via the cells of a hyperplane arrangement.
To define it, let $e_1,\dots,e_d$ denote the standard unit vectors in $\mathbb{R}^d$, and consider the set $E_{FK}$ of all these vectors plus ``diagonals'':
\[E_{FK} = \{e_1,\ldots,e_d\} \cup \{u_{i,j} = e_j - e_i\ |\ 1 \leq i < j \leq d \}.\]
Then:
\begin{definition}
The \textbf{Freudenthal--Kuhn triangulation} is induced by the hyperplane arrangement
\[\mathcal{H}_{E_{FK}} =\{ x \in \mathbb{R}^d \mid \text{there exists }u\in E_{FK} \text{ s.t. }\langle x, u \rangle
\in \mathbb{Z} \}.\]
The image of the Freudenthal--Kuhn triangulation under an affine transformation (of Euclidean space) is called a \textbf{Coxeter--Freudenthal--Kuhn triangulation}. 
\end{definition} 

As suggested by its name, the cells of this hyperplane arrangement are simplices. 

One can also define the Freudenthal--Kuhn triangulation by subdividing $\mathbb{R}^d$ into unit cubes, and in turn subdivide each of these by considering the fractional parts of the coordinates of each point inside it and sorting them in an ascending order. The sorting yields a permutation of $\{ 1, \dots, d\}$\footnote{As an example, the fractional parts of the point $(1.2, 2.3)$ are $(0.2, 0.3)$. Since $0.2<0.3$, the coordinates are already sorted. Thus the point $(1.2, 2.3)$ will be assigned the identity permutation. The point $(1.3, 2.2)$ will be assigned the transposition.}; points with the same permutation are assigned to the same simplex.

In yet other words, each simplex in the unit cube contains the vertex at the `bottom left' corner of the unit cube (i.e., the vertex associated to the point $(0,\dots,0)$) and the vertices $e_{\sigma(1)}, e_{\sigma(1)}+e_{\sigma(2)},\dots,e_{\sigma(1)}+\dots+e_{\sigma(d)}$, for some permutation $\sigma\in S_d$.

As a consequence, each unit cube is subdivided into $d!$ simplices, all sharing the diagonal of the unit cube as an edge. 
{Moreover, any two simplices are congruent via a permutation matrix and a translation. Explicitly, the simplex
\begin{align*}
    T_{\tau} = \{0, e_{\tau(1)}, e_{\tau(1)}+e_{\tau(2)},\dots,e_{\tau(1)}+\dots+e_{\tau(d)}\}
\end{align*}
is the image of the simplex
\begin{align*}
    T_{\sigma} = \{0, e_{\sigma(1)}, e_{\sigma(1)}+e_{\sigma(2)},\dots,e_{\sigma(1)}+\dots+e_{\sigma(d)}\}
\end{align*}
under the permutation matrix related to the permutation $\tau \circ \sigma^{-1}$.}
We illustrate the subdivision of unit cubes in $\mathbb{R}^2$ and $\mathbb{R}^3$ in Figure~\ref{fig:fr_simplices}.

\begin{figure}[h!]
	\centering
	\begin{subfigure}[b]{0.35\textwidth}
		\centering
		\includegraphics[width=.8\textwidth]{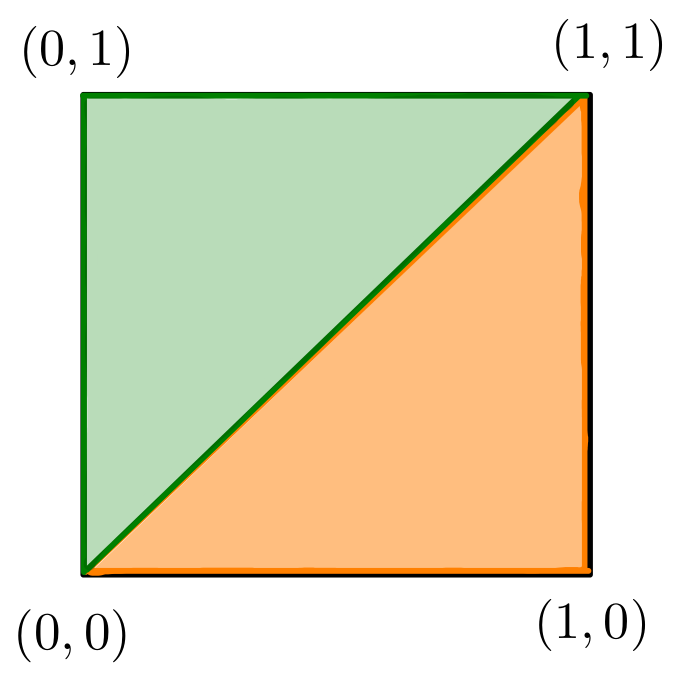}
	\end{subfigure}
	\hfill
	\begin{subfigure}[b]{0.5\textwidth}
		\centering
		\includegraphics[width=.8\textwidth]{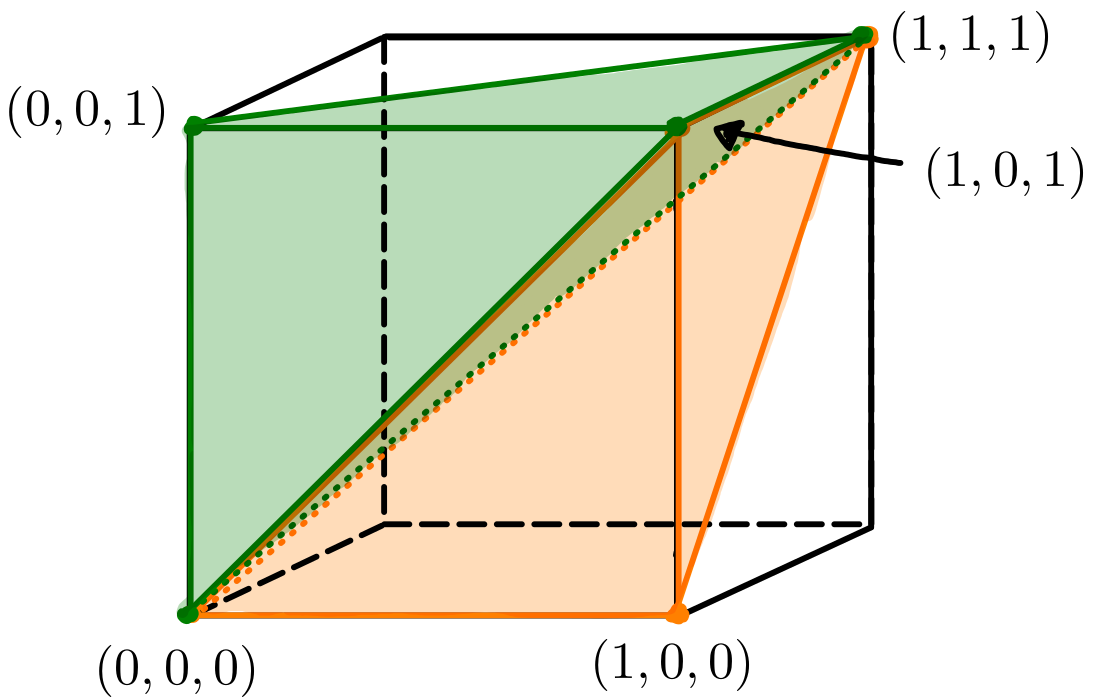}
	\end{subfigure}
	\hfill
	\caption{
		Simplices inside a unit cell in $\R^2$ (on the left) and $\R^3$ (on the right). On the right, we depict only two of the six tetrahedra: those corresponding to $x\geq z\geq y$ in orange, and $z\geq x\geq y$ in green.
	}
	\label{fig:fr_simplices}
\end{figure}

The Freudenthal--Kuhn triangulation can be rescaled by a factor $\tfrac{1}{k}$ with $k$ being an integer.  This rescaled triangulation induces a subdivision on each simplex of the original triangulation --- in dimension two we illustrate this in Figure~\ref{fig:subdivision}. We call this a \emph{subdivision by a factor of $k$}. This subdivision was also used in \cite{edelsbrunner1999edgewise}.

\begin{figure}[h!]
	\centering
		\includegraphics[width=.3\textwidth]{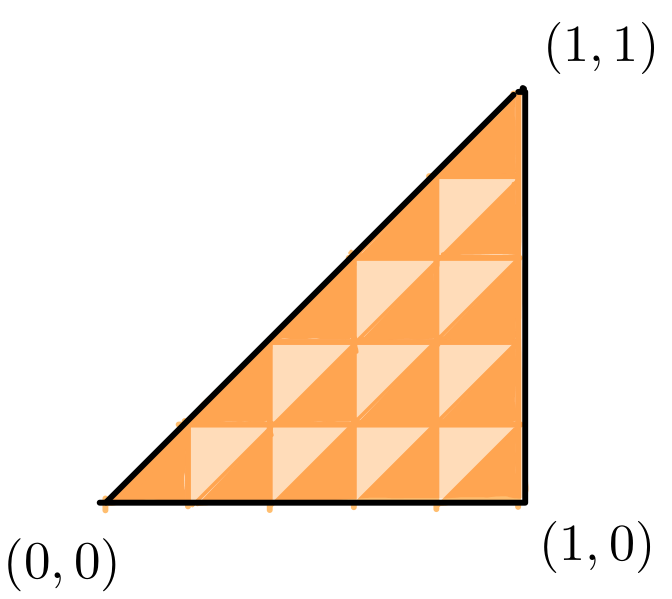}
	\caption{
		Subdivision by a factor of 5 of a simplex in the Freudenthal--Kuhn triangulation of $\R^2$.
	}
	\label{fig:subdivision}
\end{figure}

\begin{theorem}[\cite{freudenthal1942simplizialzerlegungen}]\label{thm:Freudenthal}
    Consider a subdivision by a factor $k$ of a simplex $\Delta$ in the Freudenthal--Kuhn triangulation. Then each simplex in the subdivision is similar to $\Delta$, related by a scaling factor of $\tfrac{1}{k}$. Thus, the fatness of simplices in the subdivision equals the fatness of $\Delta$. As a consequence, it is independent of $k$.
\end{theorem}
Proving this statement (for $k =2^d$) was the main goal of Freudenthal's work in~\cite{freudenthal1942simplizialzerlegungen}.

\paragraph{Subdivision of a Euclidean Simplex}
The subdivision of the Freudenthal--Kuhn triangulation can be used to induce a subdivision of any Euclidean simplex using barycentric coordinates. We construct one such map explicitly. Consider a simplex $\Delta = \{v_0,\dots,v_d\}\subseteq \R^d$. Assume without loss of generality that $v_0=0$. Then $\Delta$ can be seen as the image of the simplex $T = \{0, e_1, e_1+e_2,\dots, e_1+\cdots+e_d\}$ under the linear transformation
\[
M:\R^d\to \R^d,\qquad e_i\mapsto v_i-v_{i-1}.
\]
This map satisfies $M(e_1+\cdots+e_i)=v_i$, and, for an integer factor $k\geq 2$, the rescaled map 
\[
\tfrac{1}{k}M:k T\to \Delta, \qquad x\mapsto \tfrac{1}{k}M(x),
\]
induces a triangulation on the simplex $\Delta$ as an image of the Freudenthal--Kuhn triangulation of the simplex $k T$.


{Unfortunately, the simplices in such a subdivision do not in general preserve the fatness of the original simplex.

\begin{lemma}~\label{lemma:example_fatness_loss}
    Let $\Delta\subseteq \R^3$ be the regular tetrahedron with edge length 1. 
    Let $M:\R^3\to\R^3$ be any affine transformation mapping one simplex of the Freudenthal--Kuhn triangulation onto $\Delta$. Consider a subdivision of $\Delta$ induced by $M$. Then there are always simplices in this subdivision that are of lower fatness than $\Delta$.  
\end{lemma}
There is a geometric reason for this. The regular simplex, in dimension strictly greater than $2$, does not tile the Euclidean space. Thus, there must be simplices in a subdivision (induced by the Freudenthal--Kuhn triangulation) of the regular simplex that are not regular.
At the same time, the regular simplex has the best possible fatness among simplices. Thus, the fatness of the non-regular simplices in the subdivision has to be lower. The proof of Lemma~\ref{lemma:example_fatness_loss} can be found in Appendix~\ref{sec:proof_fatness_loss}.


While we need to account for some loss of fatness when subdividing a Euclidean simplex using the Freudenthal--Kuhn triangulation, this loss can be easily controlled, and is independent of the subdivision factor.

Indeed, let $\Delta\subseteq\R^d$ be a Euclidean simplex subdivided by a factor of $k\geq 2$, and let $e_\Delta$ be its longest edge. Assuming without loss of generality that one of the vertices of $e_\Delta$ is the origin, let $M:\R^d\to \R^d$ be a linear map mapping the simplex $\{0, e_1,e_1+e_2,\dots,e_1+\cdots +e_d\}$ onto $\Delta$, such that the edge $e_\Delta$ is the image of the spacial diagonal $e_1+\cdots +e_d$.  
Recall that the map $M$ induces a subdivision by a factor of $k$ of the simplex $\Delta$ as the image of the Freudenthal--Kuhn triangulation of the simplex $\{0, ke_1,k(e_1+e_2),\dots,k(e_1+\cdots +e_d)\}$ under the map $\tfrac{1}{k}M$.

\begin{lemma}\label{lem:fatness_bounds}
    Let --- as above --- $\Delta\subseteq\R^d$ be a Euclidean simplex subdivided by a factor of $k\geq 2$, and let $e_\Delta$ be its longest edge. Then the volume of every simplex in the subdivision induced by the linear transformation $M$ equals $\frac{\vol(\Delta)}{k^d}$, and the length of its longest edge is upper bounded by $\tfrac{1}{k}||M||\:||M^{-1}||\length(e_\Delta)$. Thus, its fatness is lower bounded by  
    \begin{align*}
       \fatness(M(\Delta))\geq\frac{1}{||M||^d\:||M^{-1}||^d} \cdot\frac{\vol(\Delta)}{\length(e_\Delta)^d}.
    \end{align*}
    In particular, it is independent of the subdivision factor $k \geq 2$.
\end{lemma}

The factor $||M||\:||M^{-1}||$ in the statement above is called the condition number of the map $M$. Roughly speaking, it measures how `skewed' $M$ is. The fact that the bound in the statement is independent of the subdivision factor $k \geq 2$ follows directly from Freudenthal's Theorem~\ref{thm:Freudenthal}.

\begin{proof}
Any simplex in the Freudenthal--Kuhn triangulation is congruent to the simplex $\{0, e_1,e_1+e_2,\dots,e_1+\cdots +e_d\}$, and thus the volume of the image of such a simplex under the map $\tfrac{1}{k}M$ equals
\[
\det \tfrac{1}{k}M\cdot\vol(\{0, e_1,e_1+e_2,\dots,e_1+\cdots +e_d\}) = \frac{\det M}{k^dd!} = \frac{\vol(\Delta)}{k^d}.
\]
Let $e_T$ be the longest edge of a simplex in the subdivision of $\Delta$. We use the fact that the edge $M^{-1}(e_T)$ is shorter or equal to the length of the diagonal of the cube with edge $\tfrac{1}{k}$. With that,
\begin{align*}
    \length(e_T)&\leq ||M||\cdot\length(M^{-1}(e_T))\\
    &\leq \tfrac{1}{k} ||M||\cdot ||e_1+\cdots +e_d|| \\
    &=\tfrac{1}{k} ||M||\cdot ||M^{-1} (M(e_1+\cdots +e_d))||\\
    &\leq \tfrac{1}{k} ||M||\: ||M^{-1}|| \length(e_\Delta).
\end{align*}
\end{proof}
}


\paragraph{Models for Simply Connected Spaces of Constant Curvature and Radial Projection} 

We use the following two models for simply connected spaces of constant non-zero curvature: For spaces of \emph{positive} constant curvature $\kappa>0$ we use the \emph{symmetrically embedded sphere} in $\mathbb{R}^{d+1}$:
\[\S_\kappa^d = \{x\in\R^{d+1}\mid \langle x,x\rangle=\tfrac{1}{\kappa}\},\]
where $\langle \cdot,\cdot\rangle$ indicates the Euclidean inner product.
For spaces of \emph{negative} constant curvature $\kappa<0$ we use the \emph{hyperboloid model} in the Minkowski space $\mathbb{R}^{d,1}$:
\[\H_\kappa^d = \{x\in\R^{d,1}\mid \langle x,x\rangle=\tfrac{1}{\kappa}\},\]
where in this case $\langle \cdot, \cdot\rangle$ indicates the Minkowski inner product $\langle x,y\rangle = \sum_{i=1}^d x_iy_i - x_{d+1}\:y_{d+1}$.
Both models are depicted in Figure~\ref{fig:radial_projection}.

\begin{figure}[h!]
	\centering
		\includegraphics[width=\textwidth]{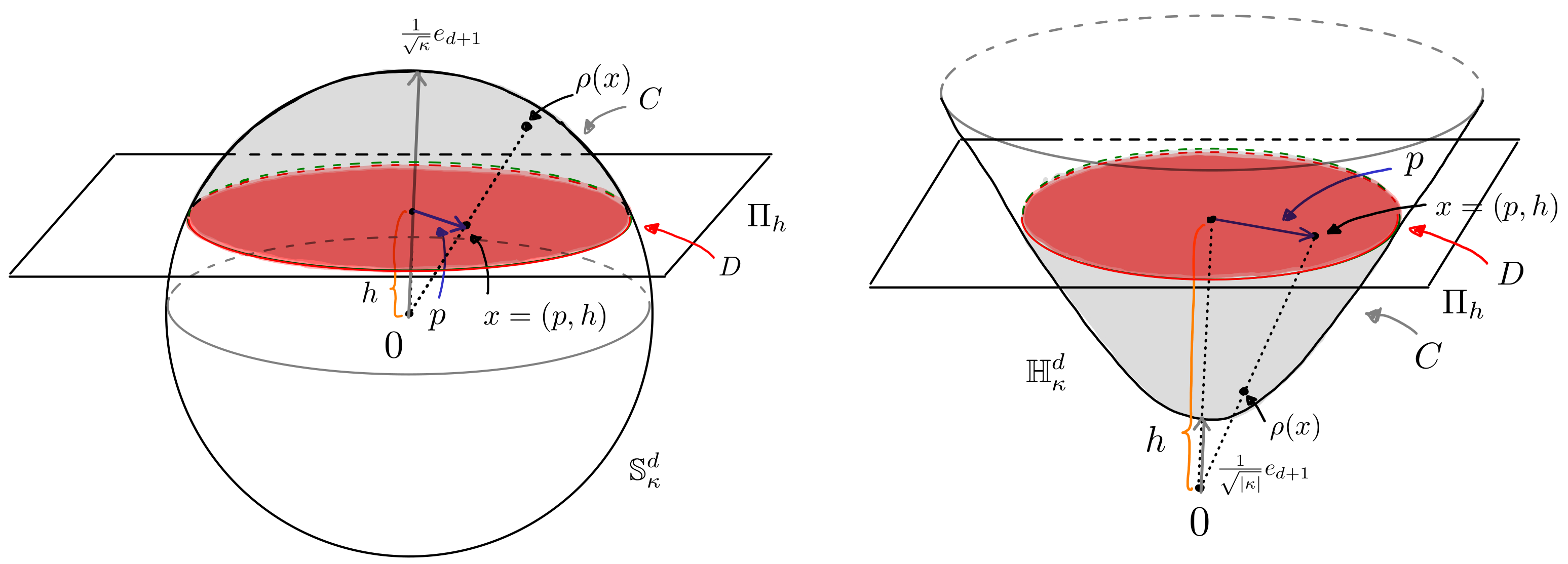}
	\caption{
		The two models, and the domain (in red) and target space (in gray) of the radial projection.
	}
	\label{fig:radial_projection}
\end{figure}

Let $\Pi_h$ be a hyperplane of $\mathbb{R}^{d+1}$ (resp. $\mathbb{R}^{d,1}$) perpendicular to the last unit vector $e_{d+1}$, that intersects $\S_k^d$ (resp. $\H_k^d$) at height $h$:
\begin{align}\label{eq:Pi_definition}
    \Pi_h = h\cdot e_{d+1} + \{x\in\R^{d+1} \mid \langle x, e_{d+1}\rangle = 0\}.
\end{align}

Let $D\subseteq \Pi_h$ be the disc whose boundary is the intersection of $\Pi_h$ with $\S_\kappa^d$ (resp.~$\H_\kappa^d$):

\[D= \Pi_h\cap \begin{cases} \{x\in \R^{d+1}\mid \langle x,x \rangle \leq \tfrac{1}{\kappa}\}& \text{if }\kappa>0,\\
\{x\in \R^{d,1}\mid \langle x,x \rangle \leq \tfrac{1}{\kappa}\}& \text{if }\kappa<0.
    \end{cases}
\]
It will sometimes be convenient to write the elements of $\Pi_h$ as
\[\Pi_h\ni x = (p,h), \qquad \text{with } p\in\R^d.\]
Points $x=(p,h)\in D$ can then be viewed as those with $||p||^2+h^2\leq \tfrac{1}{\kappa}$ if $\kappa>0$, or those with $||p||^2-h^2\leq \tfrac{1}{\kappa}$ if $\kappa<0$.

Let $C$ be the (intrinsically) convex cap in $\S_\kappa^d$ (resp. $\H_\kappa^d$) whose boundary is the intersection of $\Pi_h$ with $\S_\kappa^d$ (resp. $\H_\kappa^d$):
\[C=
\begin{cases} \S_\kappa^d\cap \bigcup_{t\geq h}\Pi_t,  &\text{if }\kappa>0,\\
\H_\kappa^d\cap \bigcup_{t\leq h}\Pi_t, & \text{if }\kappa<0.
    \end{cases}
\]

Then for each point $x\in D$, the ray through $x$ emanating from the origin intersects $C$ in a unique point, and vice versa. This bijection between $D$ and $C$ is called the \emph{radial projection}, and we denote it by $\rho$:
\begin{align}\label{eq:radial_projection}
    \rho:D\to C, \qquad x = (p,h)\mapsto \begin{cases}
        \tfrac{x}{\sqrt{\kappa(||p||^2+h^2)}}, & \text{if }\kappa>0,\\
        \tfrac{x}{\sqrt{\kappa(||p||^2-h^2)}}, & \text{if }\kappa<0.
    \end{cases}
\end{align}

Since, in both models, geodesics are the intersections of two-dimensional planes through the origin in the ambient space and the model, the radial projection of an edge in $D$ yields a geodesic edge in $C$.
Similarly, given a set $S\subseteq D$, the radial projection of its convex hull (in $D$) is the convex hull (in $C$) of the radial projection of $S$. Hence, the radial projection of a Euclidean simplex in $D$ yields a simplex in $C$.

\begin{lemma}
    Let $(g_{ij})$ be the Riemannian metric induced by the radial projection $\rho$. Then the volume distortion of $\rho$, $\vdis(\rho) = \sqrt{\det(g_{ij})}$, is bounded by
    \begin{align}\label{eq:bounds_rho}
        \begin{cases}
        \sqrt{\kappa}h \:dx_D  \leq \vdis(\rho) dx_C \leq \tfrac{1}{(\sqrt{\kappa} h)^d}dx_D, & \text{if }\kappa>0,\\
        \sqrt{|\kappa|}h \:dx_D \geq\vdis(\rho) dx_C \geq \tfrac{1}{(\sqrt{|\kappa|} h)^d}dx_D, & \text{if }\kappa<0,
    \end{cases}
    \end{align}
    where $dx_S$ denotes the volume form on the space $S$.
    
     The metric distortion of $\rho$ is upper bounded by
      \begin{align}
        \begin{cases}\label{eq:bounds_metric}
        \tfrac{\sqrt{2}}{\sqrt{\kappa} h}, & \text{if }\kappa>0,\\
        \sqrt{|\kappa|} h, & \text{if }\kappa<0.
    \end{cases}
    \end{align}
\end{lemma}

\begin{corollary}
    Let $(\tilde{g}_{ij})$ be the Riemannian metric induced by the inverse map of the radial projection $\rho^{-1}$. Then the volume distortion $\vdis(\rho^{-1})=\sqrt{\det(\tilde{g}_{ij})}$ of $\rho^{-1}$ is bounded by
    \begin{align}\label{eq:bounds_rho_inverse}
        \begin{cases}
        (\sqrt{\kappa}h)^d \:dx_C  \leq \vdis(\rho^{-1}) dx_D \leq \tfrac{1}{\sqrt{\kappa} h}dx_C, & \text{if }\kappa>0,\\
        (\sqrt{|\kappa|}h)^d \:dx_C \geq\vdis(\rho^{-1}) dx_D \geq \tfrac{1}{\sqrt{|\kappa|} h}dx_C, & \text{if }\kappa<0.
    \end{cases}
    \end{align}
\end{corollary}

\begin{remark}
    The number $\tfrac{1}{\sqrt{|\kappa|}}$ is precisely the (Euclidean) distance from the origin to any point at the boundary of the disc $D$. If one considers the cone through $D$ with apex at the origin (in other words, the conical hull of $D$), the expression $\sqrt{|\kappa|}h = \tfrac{h}{\tfrac{1}{\sqrt{|\kappa|}}}$ equals the cosine of the cone angle of this cone.
\end{remark}

\begin{proof}
    To derive bounds on the metric and volume distortion of the radial projection, we use standard techniques from differential geometry.

 The proof consists of three parts. In Part One, we calculate the coefficients of the Riemannian metric in the coordinate system on $C$ induced by $\rho$. In Part Two, we use these coefficients to determine the bounds on the volume distortion. Finally, in Part Three we use the results from Part One to determine the metric bounds.

\paragraph{Part One:} Choose a point $x=(p,h)\in D$. We denote the $i$-th coordinate vector of $\R^{d+1}$ by $e_i$ and the $i$-th coordinate of $p$ by $p_i$. Then, for $i=1,\dots,d$,
\begin{align*}
    \partial_i \rho(x)=\tfrac{\partial\rho}{\partial p_i}(p,h) = \begin{cases}
        \tfrac{1}{\sqrt{\kappa}(||p||^2 + h^2)^{3/2}}[(||p||^2 + h^2) e_i -p_i (p,h)], & \text{if }\kappa>0,\\
        \tfrac{-1}{\sqrt{|\kappa|}(-||p||^2 + h^2)^{3/2}}[(||p||^2 - h^2) e_i - p_i (p,h)], & \text{if }\kappa<0.
    \end{cases}
\end{align*}

The Riemannian metric at a point $x=(p,h)$ then equals
\[g_{ij} = \langle \partial_i \rho(x),\partial_j \rho(x) \rangle =\begin{cases}
    \tfrac{1}{\kappa(||p||^2 + h^2)^2}[(||p||^2 + h^2)\delta_{ij} - p_i p_j]& \text{if }\kappa>0,\\
    \tfrac{-1}{|\kappa|(-||p||^2 + h^2)^2}[(||p||^2 - h^2)\delta_{ij} - p_i p_j]& \text{if }\kappa<0.
\end{cases}\]

\paragraph{Part Two:} The volume element of the Riemannian metric is given by $\sqrt{\det(g_{ij})} dp_1 \wedge \dots \wedge dp_d$, meaning that the volume distortion $\vdis(\rho)$ is governed by the determinant $\det(g_{ij})$. 

Let us first consider the case $\kappa>0$. By setting $\alpha = ||p||^2 + h^2$, we obtain
\begin{align*}
    (g_{ij})=\tfrac{1}{\kappa\alpha^2}(\alpha \cdot I - pp^T).
\end{align*}
With that,
\begin{align*}
    \det(g_{ij}) &= \tfrac{1}{\kappa^d\alpha^{2d}}\underbrace{\det(\alpha \cdot I - pp^T)}_{=\alpha^{d-1}(\alpha - ||p||^2)}=\frac{h^2}{\kappa^d \alpha^{d+1}}= \frac{1}{\kappa^d \cdot (||p||^2 + h^2)^d}\cdot\frac{h^2}{||p||^2 + h^2}.
\end{align*}

Because $x=(p,h)$ lies in $D$, we know that $h^2\leq ||p||^2+h^2 \leq \tfrac{1}{\kappa}$. 
Thus, $\vdis(\rho)$ satisfies:
\begin{align*}
     \sqrt{\kappa}h \leq \vdis(\rho) \leq \tfrac{1}{(\sqrt{\kappa} h)^d},
\end{align*}
which is precisely the first part of Eq.~\ref{eq:bounds_rho}.

If $\kappa<0$, we substitute $\alpha = ||p||^2 - h^2$ and follow the same calculation to obtain
\begin{align*}
    \det(g_{ij}) &=\frac{1}{\kappa^d \cdot (||p||^2 - h^2)^d}\cdot\frac{-h^2}{||p||^2 - h^2}\\
    &=\frac{1}{|\kappa|^d \cdot (-||p||^2 + h^2)^d}\cdot\frac{h^2}{-||p||^2 + h^2}.
\end{align*}

Similarly, because $x=(p,h)$ lies in $D$, we know that $h^2\geq-||p||^2+h^2\geq \tfrac{1}{|\kappa|}$. 
Thus, $\vdis(\rho)$ satisfies:
\begin{align*}
     \sqrt{|\kappa|}h \geq \vdis(\rho) \geq \tfrac{1}{(\sqrt{|\kappa|} h)^d},
\end{align*}
which is the second part of Eq.~\ref{eq:bounds_rho}.

\paragraph{Part Three:} 
We will use the same inequalities to upper bound the metric distortion of the radial projection $\rho$.
The infinitesimal line element of the Riemannian metric induced by the radial projection is given by $\sqrt {\sum_{ij} g_{ij} dp_i dp_j }$. This can be estimated as follows:
\begin{align*}
    \sum_{ij} g_{ij} dp_i dp_j
    \leq \Big |\sum_{ij} g_{ij} dp_i dp_j \Big |
    \leq \sum_{ij} |g_{ij}| dp_i dp_j 
    \leq \max_{i,j}|g_{ij}|\sum_{ij} dp_i dp_j.
\end{align*}

Let $\kappa>0$. Then:
\begin{align*}
    \max_{i,j}|g_{ij}(p,h)|&\leq \tfrac{1}{\kappa(||p||^2 + h^2)^2}[||p||^2 + h^2 + \max_{i,j}|p_i p_j|]\leq \tfrac{1}{\kappa(||p||^2 + h^2)^2} [2||p||^2 + h^2]\\
    &\leq \tfrac{1}{\kappa(||p||^2 + h^2)^2} 2[||p||^2 + h^2]=\tfrac{2}{\kappa(||p||^2 + h^2)}\\
    &\leq \tfrac{2}{\kappa h^2},
\end{align*}
which is precisely the first part of Eq.~\ref{eq:bounds_metric}.

For $\kappa<0$,
\begin{align*}
    \max_{i,j}|g_{ij}(p,h)|&\leq \tfrac{1}{|\kappa|(-||p||^2 + h^2)^2}[-||p||^2 + h^2 + \max_{i,j}|p_i p_j|]\\
    &\leq \tfrac{1}{|\kappa|(-||p||^2 + h^2)^2}[-||p||^2 + h^2 + ||p||^2]=\tfrac{h^2}{|\kappa|(-||p||^2 + h^2)^2}\\
    &\leq |\kappa| h^2,
\end{align*}
which is the second part of Eq.~\ref{eq:bounds_metric}.
\end{proof}

\section{Subdivision of Simplices of Constant Curvature} \label{sec:mainConstruction}
We are now ready to prove our claim:
\begin{theorem}\label{thm:main}
Any simplex $\Delta$ of constant non-zero curvature, whose circumradius is either finite (if the curvature is negative) or is less than the circumradius of the model sphere (if the curvature is positive), has a subdivision of simplices whose fatness is lower bounded in terms of the fatness and circumradius of $\Delta$.
\end{theorem}

\begin{remark}\label{remark:circumradius_dependence}
    The dependence of the bound on the circumradius of the simplex is natural. Indeed, when the circumradius of a simplex is infinite in negatively curved ambient spaces, or equal to that of the model sphere in positively curved ambient spaces, the radial projection defined in Eq.~\eqref{eq:radial_projection} is no longer well defined. Accordingly, in these cases the bounds in Eq.~\eqref{eq:bounds_fatness} degenerate.

    When the curvature of the ambient space is positive, and the circumradius of a simplex --- seen as a collection of points --- equals the radius of the model sphere, the convex hull of these points is not well-defined.
    
    When the curvature of the ambient space is negative, it might happen that the simplex is well-defined and finite, while its circumradius is infinite. In this case, our approach fails to provide appropriate bounds on the simplex quality. We suggest an alternative approach to lower-bound the quality in Section~\ref{sec:final_remarks}.
\end{remark}

\begin{proof}
Let $d$ be the dimension of $\Delta$, $\kappa\neq 0$ its curvature, and $r$ its circumradius.
Further, let 
\begin{align}\label{eq:h}
    h = \begin{cases}
    \tfrac{1}{\sqrt{\kappa}}\cos(\sqrt{\kappa}r) & \text{if }\kappa>0,\\
    \tfrac{1}{\sqrt{|\kappa|}}\cosh(\sqrt{|\kappa|}r) & \text{if }\kappa<0.
\end{cases}
\end{align}
Then $\Delta$ can be embedded in the corresponding model of space of constant curvature such that its vertices lie in the hyperplane $\Pi_h$ as defined in~\eqref{eq:Pi_definition}. We subdivide $\Delta$ as illustrated in Figure~\ref{fig:fr_subdivision}:

The hyperplane $\Pi_h$ inherits a Euclidean metric, both when its ambient space is $\R^{d+1}$ and when it is $\R^{d,1}$. Thus, the convex hull of the vertices of $\Delta$ in $\Pi_h$ yields a Euclidean simplex. This simplex can be subdivided by a factor $k$ according to Freudenthal's scheme (in red in Figure~\ref{fig:fr_subdivision}). Finally, the radial projection of this subdivision onto the space of constant curvature yields a subdivision of $\Delta$ (in green in Figure~\ref{fig:fr_subdivision}).
\begin{figure}[h!]
	\centering
	\begin{subfigure}[b]{0.4\textwidth}
		\centering
		\includegraphics[width=.9\textwidth]{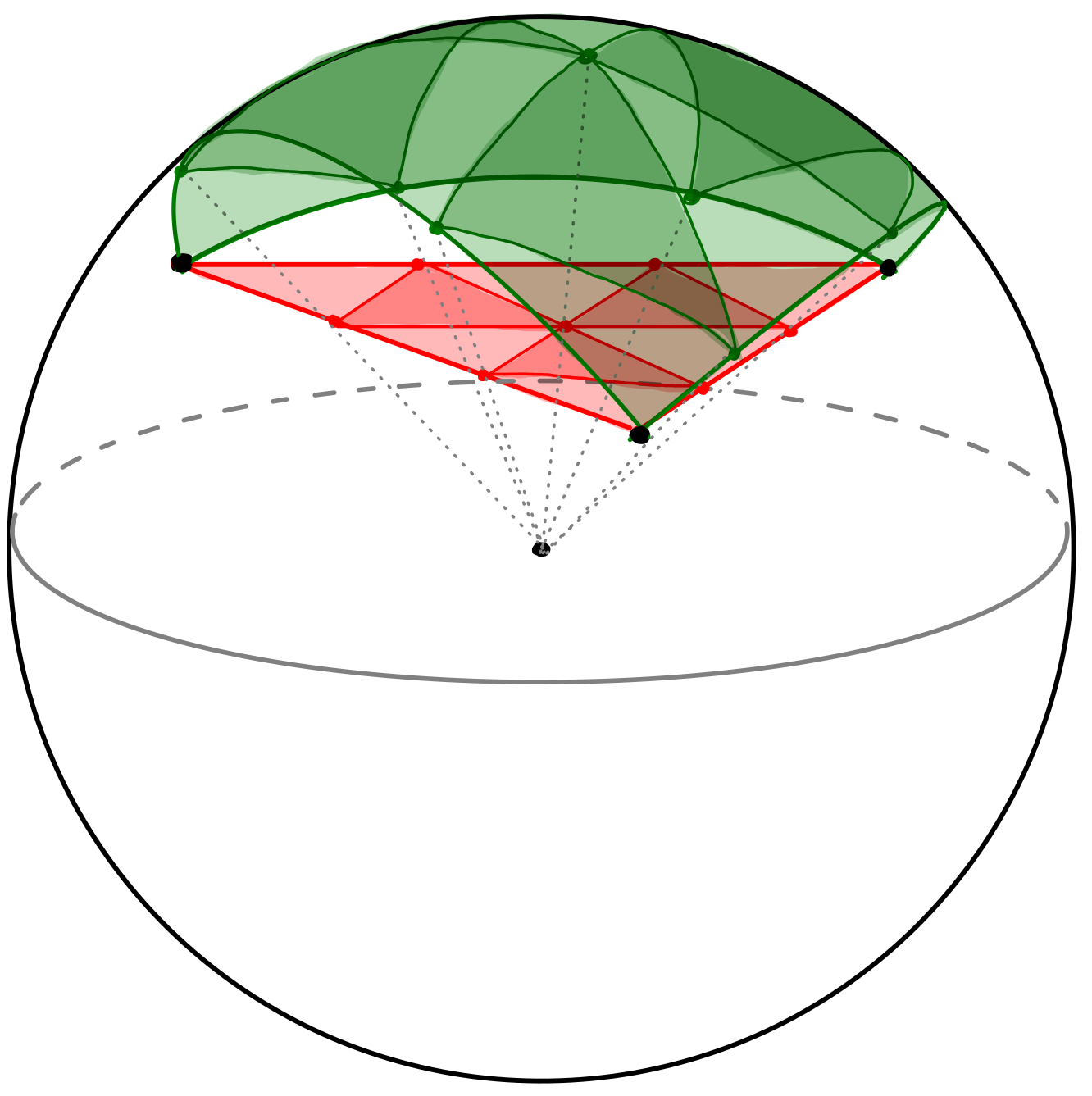}
	\end{subfigure}
	\qquad\quad
	\begin{subfigure}[b]{0.4\textwidth}
		\centering
		\includegraphics[width=.9\textwidth]{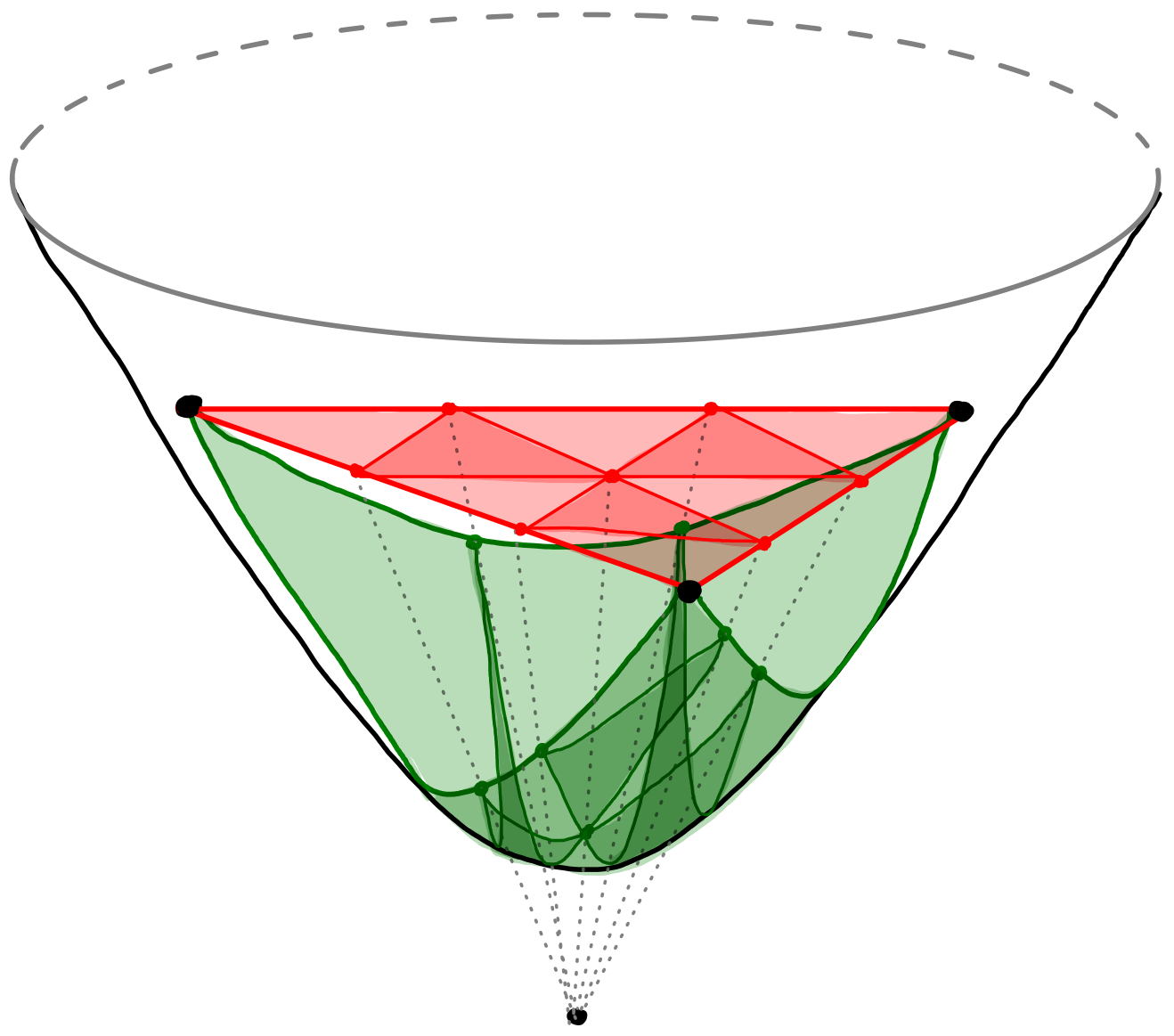}
	\end{subfigure}
	\hfill
	\caption{
		Radial projection of Freudenthal's scheme for $d=2$ and $k=3$. The simplex $\Delta$ is in green, the corresponding Euclidean simplex in red.
	}
	\label{fig:fr_subdivision}
\end{figure}

Let $M$ be a map inducing the subdivision of the Euclidean simplex $\rho^{-1}(\Delta)$, and let $T$ be a simplex in the image of this subdivision, that is, in the subdivision of $\Delta$. We first consider the case where $\kappa>0$:
\begin{enumerate}
    \item Due to the bound on the volume distortion of $\rho^{-1}$ in Eq.~\eqref{eq:bounds_rho_inverse}, $\vol(T)\geq \sqrt{\kappa}h\vol(\rho^{-1}(T))$.
    \item Due to Lemma~\ref{lem:fatness_bounds}, $\vol(\rho^{-1}(T)) = \tfrac{1}{k^d}\vol(\rho^{-1}(\Delta))$.
    \item Finally, due to the bound on the volume distortion of $\rho$ in Eq.~\eqref{eq:bounds_rho}, $\vol(\rho^{-1}(\Delta))\geq (\sqrt{\kappa} h)^d\vol(\Delta)$.
\end{enumerate}
Thus, 
\[\vol(T)\geq\tfrac{1}{k^d}\cdot(\sqrt{\kappa} h)^{d+1}\vol(\Delta).\]
When $\kappa<0$, the same argument yields
\[\vol(T)\geq\tfrac{1}{k^d}\cdot\tfrac{1}{(\sqrt{|\kappa|} h)^{d+1}}\vol(\Delta).\]

Next, let $e_T$ be the longest edge of the simplex $T$, and $e_\Delta$ be the longest edge of the simplex $\Delta$. Let $\kappa>0$.
\begin{enumerate}
    \item Thanks to the metric bound~\eqref{eq:bounds_metric}, $\length(e_T)\leq \tfrac{\sqrt{2}}{\sqrt{\kappa}h}\length(\rho^{-1}(e_T))$.
    \item Let $\tilde{e}$ denote the longest edge of the Euclidean simplex $\rho^{-1}(\Delta)$. The edge $\rho^{-1}(e_T)$ might not be the longest edge of the simplex $\rho^{-1}(T)$, but it is definitely not longer than the longest edge. Combining this fact with Lemma~\ref{lem:fatness_bounds}, we obtain: $\length(\rho^{-1}(e_T))\leq \tfrac{1}{k}||M||\:||M^{-1}||\length(\tilde{e})$.
    \item Since the chordal distance between two points on a sphere is always shorter than their spherical distance, $\length(\tilde{e})\leq \length(\rho(\tilde{e}))$.
    \item Finally, since $\rho(\tilde{e})$ is an edge of the simplex $\Delta$, $\length(\rho(\tilde{e}))\leq \length(e_\Delta)$.
\end{enumerate}
Combining these results yields
\[\length(e_T)\leq\tfrac{1}{k}\cdot||M||\:||M^{-1}||\tfrac{\sqrt{2}}{\sqrt{\kappa}h}\length(e_\Delta).\]
Similarly, when $\kappa<0$,
\[\length(e_T)\leq\tfrac{1}{k}\cdot||M||\:||M^{-1}||\sqrt{\kappa}h\length(e_\Delta).\]

Thus, the fatness of the simplex $T$ is lower bounded by the fatness of the simplex $\Delta$ by:
\[
\fatness(T)\geq \mathscr{C} \cdot \fatness(\Delta),
\]
with 
 \begin{align}
       \mathscr{C} = \begin{cases}\label{eq:bounds_fatness}
        \frac{(\sqrt{\kappa} h)^{2d+1}}{2^{d/2}||M||^d||M^{-1}||^d} = \frac{\cos^{2d+1}(\sqrt{\kappa}r)}{2^{d/2}||M||^d||M^{-1}||^d}, & \text{if }\kappa>0,\\
        \frac{1}{(\sqrt{|\kappa|} h)^{2d+1}||M||^d||M^{-1}||^d} =\frac{1}{\cosh^{2d+1}(\sqrt{|\kappa|}r)||M||^d||M^{-1}||^d} , & \text{if }\kappa<0.
    \end{cases}
    \end{align}
Here, we inserted the value of $h$ from equation \eqref{eq:h}.
\end{proof} 




\section{Comparison of Subdivision Methods} \label{sec:Comp}
In \cite{brunck1}, Brunck proposed the following construction in two dimensions. We reformulate it here in a way that makes clear its extension to arbitrary dimensions. Rather than subdividing the Euclidean simplex according to Freudenthal's scheme, Brunck proceeds iteratively as follows (see also Figure~\ref{fig:brunck_subdivision}):

\begin{description} 
\item[{\bf Step 1.}] Consider the convex hull of the vertices in the ambient space, which yields a Euclidean simplex (as in the construction above). 
\item[{\bf Step 2.}] Subdivide the Euclidean simplex using Freudenthal's scheme of factor $2$. 
\item[{\bf Step 3.}] Project the Euclidean subsimplices radially to obtain subsimplices in the space of constant curvature. 
\item[{\bf Step 4.}] Apply the steps of the construction on each of the subsimplices.   
\end{description} 

\begin{figure}[h!]
	\centering
	\begin{subfigure}[b]{0.32\textwidth}
		\centering
		\includegraphics[width=\textwidth]{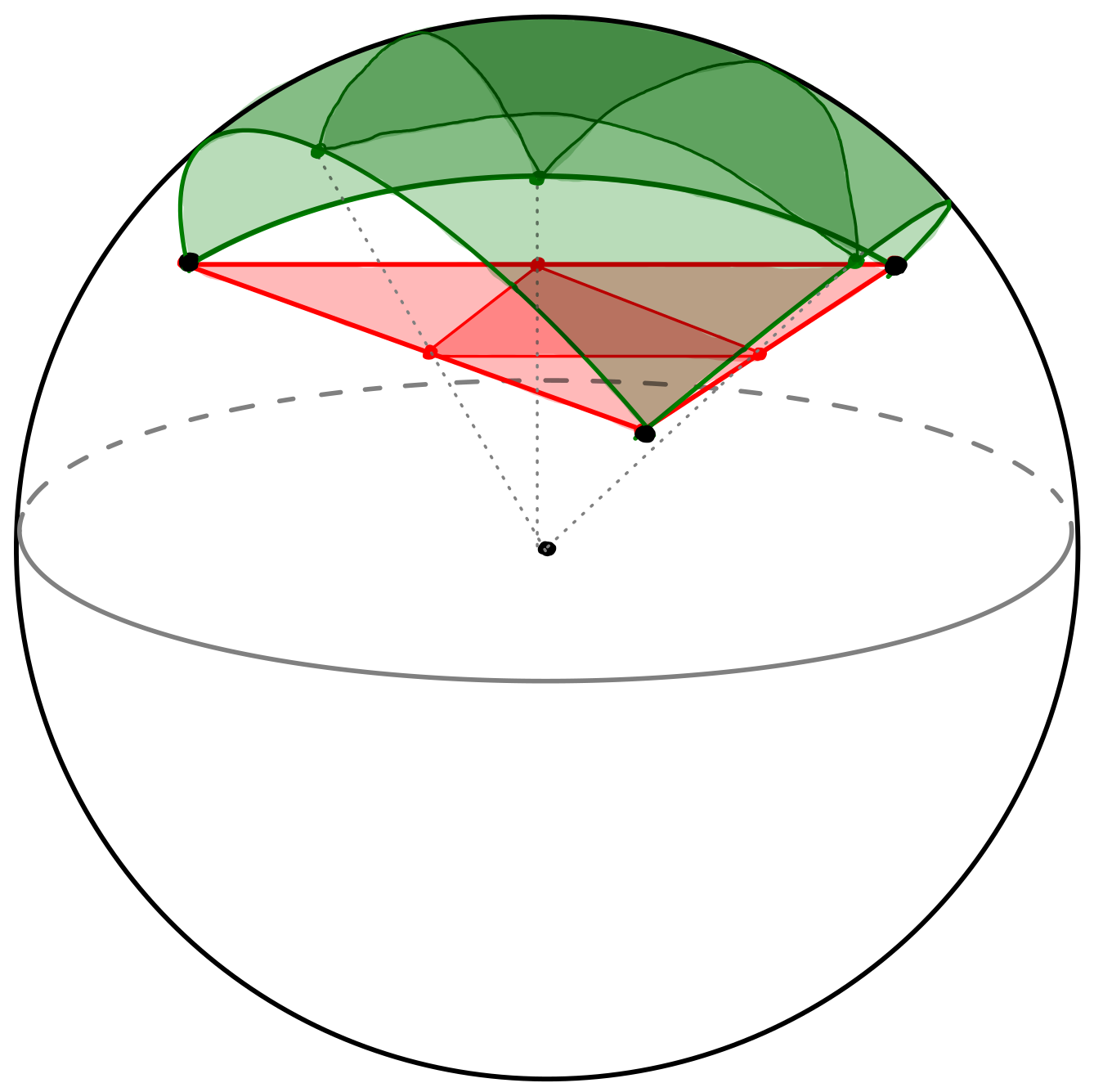}
	\end{subfigure}
	\begin{subfigure}[b]{0.32\textwidth}
		\centering
		\includegraphics[width=\textwidth]{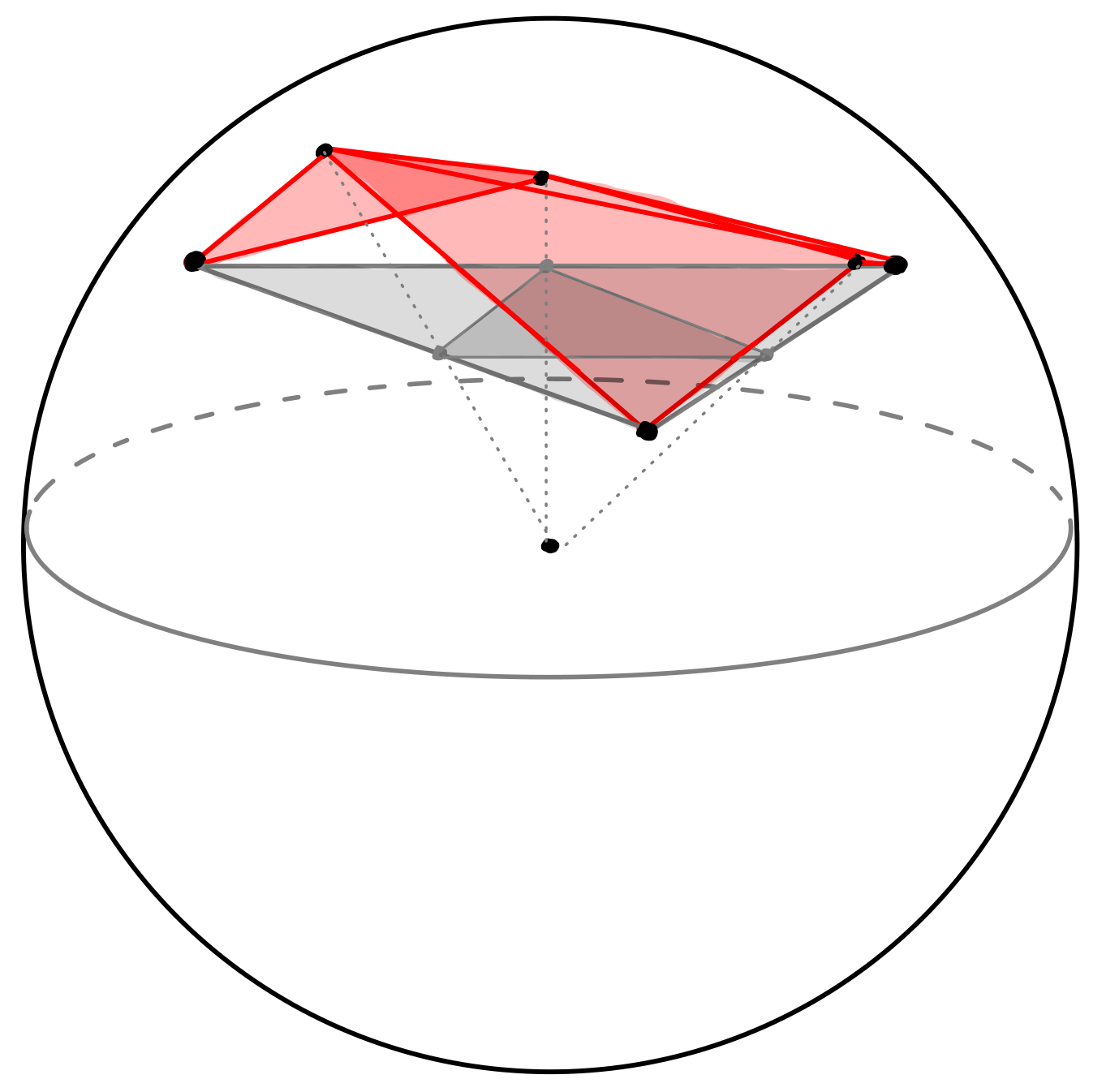}
	\end{subfigure}
	\begin{subfigure}[b]{0.32\textwidth}
		\centering
		\includegraphics[width=\textwidth]{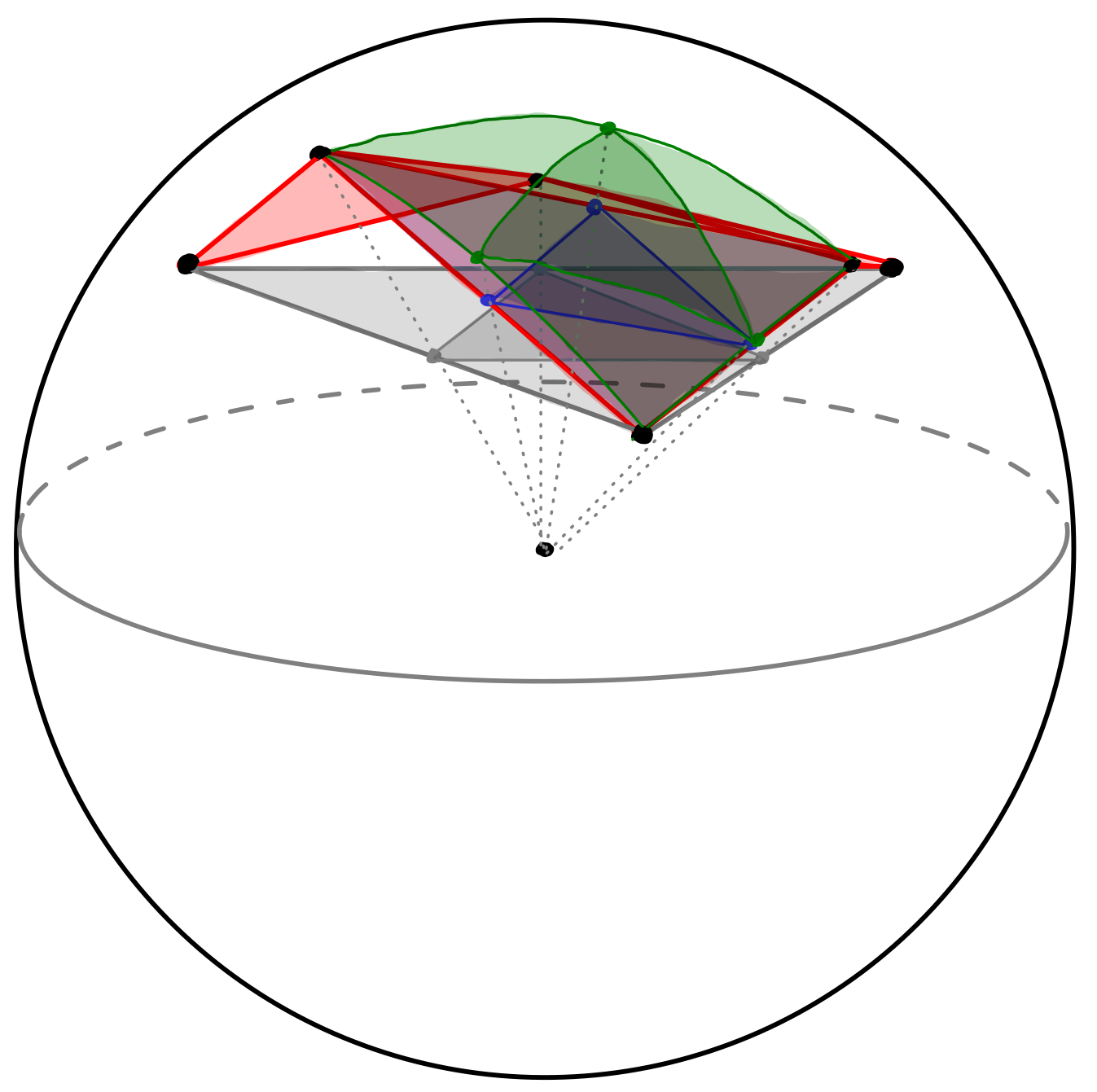}
	\end{subfigure}
	\caption{
		Left: Steps 1-3 of Brunck's subdivision scheme. Centre: The Euclidean convex hulls of subsimplices, on which the scheme is to be iterated. Right: Steps 1-3 are performed on one of the subsimplices.
	}
	\label{fig:brunck_subdivision}
\end{figure}

After $m$ iterations, this scheme yields a subdivision which is close (for small simplices) to, but different  from, the subdivision one would obtain using Freudenthal's scheme of factor $k= 2^m$. The difference stems from the following observation: If one subdivides a cord on a circle into $k$ equal parts and radially projects the pieces on the circle then the resulting parts of the circle do not all have the same length.

In contrast to this method, which alternates between subdivision by factor $k=2$ and a projection, our method projects only once, and subdivides by the desired factor immediately.

This approach has several advantages over Brunck’s. Its principal advantage is in implementation. In Brunck’s iterative method, the number of radial projections between Euclidean space and a space of constant curvature grows exponentially with each iteration. As a result, Brunck’s approach is substantially more computationally intensive and more susceptible to numerical error than ours, in which the radial projection is applied only once.

Another advantage of our approach is the flexibility to choose an arbitrary subdivision factor, rather than being restricted to powers of two. Moreover, we obtain explicit upper bounds on the edge lengths in the subdivision. Therefore, if a subdivision with a prescribed mesh size, measured by the maximum edge length, is required, the corresponding subdivision factor can be determined directly. This information is not available in Brunck’s approach and would instead have to be verified explicitly at each iteration, which is computationally expensive.

Nevertheless, Brunck’s approach may produce simplices of slightly better quality than ours. We believe, however, that any difference in quality is at most marginal.

\section{Final Remarks}\label{sec:final_remarks}
We conclude this article with several remarks and open questions arising from this work.
\paragraph{On Bound Optimization}
While the bounds in Eq.~\eqref{eq:bounds_fatness} of Theorem~\ref{thm:main} are explicit, they are not optimal and could be improved by taking the shape of the simplex into account. Such improvements could be made in two ways: first, by deriving a more refined estimate of the distortion in simplex quality under radial projection; and second, by optimizing the choice of the transformation $M$ (which is likely to depend on the dimension). We did not pursue these refinements here, since the main goal of the paper is to develop a general subdivision method.

\paragraph{An Alternative (and more General) Subdivision Scheme for Simplices in Spaces of Negative Curvature}

As noted in Remark~\ref{remark:circumradius_dependence}, our approach does not yield a subdivision scheme for simplices in negatively curved spaces with infinite circumradius, namely, simplices inscribed in horospheres or in equidistant hypersurfaces to hyperplanes\footnote{These are sometimes also called `hyperspheres'.}. Nevertheless, such simplices may still have bounded quality.

Instead of using the radial projection, which fails in these cases, one can subdivide such simplices using the Freudenthal--Kuhn triangulation in combination with the Klein--Beltrami model of hyperbolic geometry. Indeed, in the Klein--Beltrami model, hyperbolic simplices are represented by Euclidean simplices, and may therefore be subdivided directly using the Freudenthal--Kuhn triangulation. Bounds on the quality of the simplices in the resulting subdivision can then be derived from the Riemannian metric induced by the model, in much the same way as in the proofs of Lemma~\ref{lem:fatness_bounds} and Theorem~\ref{thm:main}.

The main advantage of this subdivision scheme is its greater flexibility compared with our radial-projection-based approach. In particular, it applies to every simplex in a space of constant negative curvature, including simplices with ideal vertices. On the other hand, the resulting quality bounds are likely to be significantly weaker than those obtained here.

If one chooses to use this scheme, it seems advisable to position the simplex so that its barycenter or incentre lies at the center of the ball in the Klein--Beltrami model, in order to optimize the distortion bounds.

\section*{Acknowledgements}
We are greatly indebted to Ramsay Dyer for discussion. We thank Florestan Brunck for pointing out his work and discussion. 

\bibliographystyle{plainurl}
\bibliography{biblio} 

\appendix

\section{Proof of Lemma~\ref{lemma:example_fatness_loss}}\label{sec:proof_fatness_loss}

    Without loss of generality we may assume that 
    \begin{itemize}
            \item the tetrahedron $\Delta$ has vertices
        \[v_0 = (0,0,0),\quad v_1 = (1,0,0), \quad v_2 = (\tfrac{1}{2}, \tfrac{\sqrt{3}}{2}, 0), \quad v_3=(\tfrac{1}{2},\tfrac{\sqrt{3}}{6}, \sqrt{\tfrac{2}{3}}),\]
        \item the preimage tetrahedron $M^{-1}(\Delta)$ lies in the unit cube $[0,1]^3$.
    \end{itemize}

    Since $M^{-1}(\Delta)$ lies in the unit cube, it is determined by a permutation $\sigma\in S_3$, with its vertices being
    \[
    w_0 = 0, \quad w_1 = e_{\sigma(1)}, \quad w_2 = e_{\sigma(1)}+e_{\sigma(2)}, \quad w_3 =e_{\sigma(1)}+e_{\sigma(2)}+ e_{\sigma(3)}.
    \]
    Here, $e_i$ denotes the $i$-th unit vector. Due to the symmetry of the regular simplex $\Delta$, we may in addition assume without loss of generality that
    \begin{itemize}
        \item $M(w_i)=v_i$.
    \end{itemize}
    In particular, $M$ is a linear transformation.
    
    Let $P$ be the permutation matrix related to the permutation $\sigma$. Then the map $M\circ P$ is a linear map mapping
    \[
    M\circ P(e_1) = v_1, \quad M\circ P(e_1+e_2) = v_2, \quad M\circ P(e_1+e_2+e_3) = v_3,
    \]
    and can be represented by the matrix
    \[
    M\circ P = \begin{pmatrix}
1 & -\frac{1}{2} & 0\\[4pt]
0 & \frac{\sqrt{3}}{2} & -\frac{\sqrt{3}}{3}\\[4pt]
0 & 0 & \sqrt{\frac{2}{3}}
\end{pmatrix}.
    \]
    We now analyze the fatness of the image under $M\circ P$ of each simplex in the Freudenthal--Kuhn triangulation of the unit cube. To that end, notice that all simplices in the Freudenthal--Kuhn triangulation are congruent, and thus have the same volume. Since the volume of the unit cube equals 1, and there are $3!=6$ simplices in the unit cube, each simplex has volume $\tfrac{1}{6}$. Thus, the volume of the image under $M\circ P$ of each simplex equals
    \[|\det(M\circ P)| \cdot \tfrac{1}{6} = \tfrac{1}{6\sqrt{2}}.\]

    Next, let us examine the image of the edges of the Freudenthal--Kuhn triangulation of the unit cube. 
    The edges of the triangulation are
    \begin{itemize}
        \item all edges of the cube --- the edges $e_1,e_2,e_3$, and their parallels,
        \item the diagonals $e_1+e_2, e_1+e_3, e_2+e_3$ and their parallels, and
        \item the spacial diagonal $e_1+e_2+e_3$.
    \end{itemize}
The image of the unit cube under the map $M\circ P$ is a parallelopiped. The images of the edges $e_1, e_2, e_3, e_1+e_2, e_2+e_3, e_1+e_2+e_3$ and their parallels have edge length 1, but the images of the diagonal $e_1+e_3$ and its parallel have edge length
\[
|(M\circ P)(1,-\tfrac{1}{\sqrt{3}},\sqrt{\tfrac{2}{3}})|=\sqrt{2}.
\]

Thus, any simplex in the Freudenthal--Kuhn triangulation of the unit cube with either the vertex $e_2$\footnote{Since each simplex contains the vertex $e_1+e_2+e_3$, the simplices containing the vertex $e_2$ will contain the edge connecting the vertex $e_2$ to the vertex $e_1+e_2+e_3$. This edge is parallel to $e_1+e_3$.} or the vertex $e_1+e_3$ will be mapped to a simplex with the longest edge of length $\sqrt{2}$ by the map $M\circ P$. Note that four out of the six simplices in the triangulation of the cube satisfy this condition.

Finally, recall that $P$ is a permutation matrix related to a permutation $\sigma$, and it maps $e_i$ onto $e_{\sigma(i)}$. Thus, the map $M$ maps each simplex with either the vertex $e_{\sigma^{-1}(2)}$ or the vertex $e_{\sigma^{-1}(1)}+e_{\sigma^{-1}(3)}$ to a simplex with the longest edge of length $\sqrt{2}$. These simplices have the fatness:
\[
\frac{\tfrac{1}{6\sqrt{2}}}{(\sqrt{2})^3} = \frac{1}{24}.
\]
In comparison, the original simplex $\Delta$ has the fatness
\[
\fatness(\Delta)=\frac{1}{6\sqrt{2}}.
\]

\end{document}